\documentclass[]{elsarticle} % 3p,twocolumn

\usepackage{hyperref,color}
\journal{Systems \& Control Letters}

\usepackage{amsthm}
\usepackage{amsmath} % assumes amsmath package installed
\usepackage{amssymb}  % assumes amsmath package installed
\usepackage{verbatim}
\usepackage{empheq} 
\usepackage{graphicx} % for pdf, bitmapped graphics files
\usepackage{mathptmx} % assumes new font selection scheme installed
\usepackage[]{algorithm2e}
\usepackage{subfig}

\DeclareMathOperator{\trace}{trace}
\renewcommand{\vec}[1]{\mathbf{#1}}
\newtheorem{theorem}{Theorem}
\newtheorem{proposition}{Proposition}

\usepackage[capitalise]{cleveref}

\graphicspath{{images/}}

\begin{document}

\begin{frontmatter}

\title{Augmented State Feedback for Improving Observability of\\
Linear Systems with Nonlinear Measurements}
%\tnotetext[mytitlenote]{Fully documented templates are available in the elsarticle package on \href{http://www.ctan.org/tex-archive/macros/latex/contrib/elsarticle}{CTAN}.}

%% Group authors per affiliation:
\author{Atiye~Alaeddini}
%\ead{email@uni.edu}
\address{Institute for Disease Modeling, Bellevue, WA, 98005}

\author{Kristi~A.~Morgansen}
%\address{University of Washington, William E. Boeing Department of Aeronautics and Astronautics, Seattle, WA 98195}
%\fntext[myfootnote2]{Senior~Member,~IEEE}
\author{Mehran~Mesbahi}
\address{University of Washington, William E. Boeing Department of Aeronautics and Astronautics,\\ Seattle, WA 98195}
%\fntext[myfootnote3]{Fellow,~IEEE}

%% or include affiliations in footnotes:
%\author[mymainaddress,mysecondaryaddress]{Elsevier Inc}
%\ead[url]{www.elsevier.com}
%
%\author[mysecondaryaddress]{Global Customer Service\corref{mycorrespondingauthor}}
%\cortext[mycorrespondingauthor]{Corresponding author}
%\ead{support@elsevier.com}
%
%\address[mymainaddress]{1600 John F Kennedy Boulevard, Philadelphia}
%\address[mysecondaryaddress]{360 Park Avenue South, New York}

\begin{abstract}
This paper is concerned with the design of an augmented state feedback controller for finite-dimensional linear systems with nonlinear observation dynamics. Most of the theoretical results in the area of (optimal) feedback design are based on the assumption that the state is available for measurement.
% This assumption does not hold in the real-world problems. 
In this paper, we focus on finding a feedback control that avoids state trajectories with undesirable observability properties. 
In particular, we introduce an optimal control problem that specifically considers an index of observability in the control synthesis. The resulting cost functional is a combination of LQR-like quadratic terms and an index of observability. The main contribution of the paper is presenting a control synthesis procedure that on one hand, provides closed loop asymptotic stability, and addresses the observability of the system--as a transient performance criteria--on the other.
\end{abstract}

\begin{keyword}
%\texttt{elsarticle.cls}\sep \LaTeX\sep Elsevier \sep template
%\MSC[2010] 00-01\sep  99-00
Nonlinear observations\sep linear state feedback\sep observability\sep stability
\end{keyword}

\end{frontmatter}

%\linenumbers

\section{Introduction}

A primary consideration in the control synthesis procedure examined in this
paper is system observability. 
This is motivated by the inherent coupling in the
actuation and sensing in nonlinear systems; as such, one might be able to make a system more observable by changing the control inputs \cite{Alaeddini13,hinson2013observability,alaeddini2014trajectory}. Here, we consider a particular type of nonlinear systems, namely, one with linear  state dynamics paired with a nonlinear observation dynamics. 
This is motivated by the fact that in some practical problems, nonlinearities arise from the system measurements. An example of this type of systems is a simplified particle model of a vehicle with a camera or sonar with range-only/bearing-only measurements, such as an autonomous underwater vehicle (AUV). The measurements for these vehicles are usually based on acoustic localization \cite{Quenzer14}. In this case, acoustic transducers are installed on the AUV and on a set of other beacons. The AUV has the exact global location of the beacons and can determine its range to one (Single Beacon Navigation \cite{Baccou02}), three (Long-Baseline Localization \cite{Larsen00}), or more separate beacons in order to determine its position. Another example of nonlinear observation is bearing-only measurement, e.g.,~\cite{nardone1981observability}. In this case, the vehicle only has information about its direction but not about the distance from reference points. A modeling of a vehicle with azimuth bearing, conical bearing, and depth/elevation angle measurements can be found in \cite{Hammel85}.

For linear systems, Linear Quadratic Regulation (LQR) is a well known and commonly used method for designing an optimal control law given a quadratic cost function of state and control. One potential problem may be that while the LQR problem minimizes the cost functional, it may not result in good response when the measurement equation does not facilitate the use or reconstruction of the state for feedback. 
%Particularly, here, we consider the observability of the LQR optimal control
In this paper we modify the cost functional in standard LQR and suggest an augmented control in order to improve the observability of the baseline optimal controller. As such, we show that we can avoid unobservable trajectories by augmenting decaying oscillatory terms with the  LQR-induced state feedback, without effecting the asymptotic stability and performance of the baseline controller.
%In this chapter, designing a feedback controller using optimal control problem is investigated. 
%In order to formulate this control design problem using optimal control, we must define a scalar objective which scores the long-term performance of running each candidate control policy, $\vec{u}(\vec{x})$, from each initial condition, $\vec{x}_0$, and a list of constraints that must be satisfied. Optimal control has a long history in robotics. For instance, there has been a great deal of work on the minimum-time problem for pick-and-place robotic manipulators, and the linear quadratic regulator (LQR) and linear quadratic regulator with Gaussian noise (LQG) have become essential tools for any controls engineer. With increasingly powerful computers and algorithms, the popularity of numerical optimal control has grown at an incredible pace over the last few years.

Our goal in this work is proposing a synthesis procedure for efficiently choosing the inputs of a nonlinear system to improve its overall sensing performance (measured by the observability Gramian).
This task is motivated through problems where the reconstruction of the state proves to be challenging and it becomes imperative to actively acquire, sense, and estimate it; an example
of such a problem is robot localization operating in an unknown environment, when the robot only has access to range/bearing measurements \cite{lorussi2001optimal}. 
%The range and bearing measurements are popular in robotics.
%
%The problem of active sensing has received significant attention in the robotics community in recent years. 
%
In this context, a robot can use a path that leads to more effective 
localization of itself or to  construct a map of the environment \cite{lorussi2001optimal,DeVries13}. 
For example, in \cite{DeVries13}, sampling trajectories for autonomous vehicles are selected from a finite sampling of trajectories through an exhaustive search. 
%
%Moreover, the resulting optimized path was computed by an exhaustive search. 
%
On the other hand, the work \cite{lorussi2001optimal} discusses the geometry of trajectories that achieves maximum information for a model of the rover robot.
Yu {\em et al.}~\cite{Yu11} have developed a path planning algorithm based on dynamic programming for navigation of a micro aerial vehicle (MAV) using bearing-only measurements. 
In their proposed algorithm, at each time step, one of the three possible choices of roll command--that is more observable--is selected. 

The observability-based optimal controller design presented in this work is closely related to the analysis discussed in \cite{bohm2008avoidance} and \cite{alessandretti2013model} to avoid unobservable state/output trajectories. In both works, an index of observability is added to the cost functional. The index of observability used in both of these works are inspired by the observability matrix for linear systems. The observability rank condition for nonlinear systems, given in \cite{hermann1977nonlinear,sontag1991equations}, depends on the rank of the span of the observation space of the nonlinear system. In general, the observability matrix of a nonlinear system is the span of the Lie derivatives of the output function in all possible directions generated from the drift and control vector fields. In this work, we introduce an index of observability that is closely related to the observability Gramian. Unlike the observability matrix which gives a yes/no answer to the observability question, the observability Gramian parameterizes the level of observability for linear/nonlinear system.

In comparison with previous observability-based path planning algorithms reported in the literature, the controller obtained in this paper can be applied to a larger family of nonlinear systems (with a general form of the nonlinear observation function); moreover, the obtained controller ensures closed loop stability. In particular, utilizing the concept of observability Gramian, and tools from linear optimal control theory, this paper develops a synthesis procedure for a modified version of the optimal linear state feedback with stability guarantees. The objective function examined in this paper is quadratic in state and control, augmented with an observability measure. One of our contribution is presenting an observability index, whose optimization does not have adverse effects on the asymptotic stability properties of the system.
%% ======================================================================================================= %%
\section{Preliminaries on Observability of Nonlinear Systems}
\label{sec:prelim}

Consider the linear time invariant system of the form
\begin{eqnarray}
\begin{aligned}
   \dot{\vec{x}}& = A \vec{x} + B \vec{u}  \\
   \vec{y} &= C \vec{x} \,,  \label{linear}
\end{aligned} 
\end{eqnarray}
for which the observability Gramian,
\begin{equation}
	W_{o,L} = \int_{0}^{t_f} e^{A^Tt} C^T C e^{At} \mathrm{d}t \,, \label{LinGram}
\end{equation}
can be computed to evaluate the observability of a linear system (linear dynamics and linear observation) \cite{Krener09}. 
A linear system is (fully) observable if and only if the corresponding observability Gramian is full rank \cite{muller1972analysis}. In the linear setting, if the observability Gramian is rank deficient, certain states (or directions in the state space) cannot be reconstructed from the measurement, regardless of the control policy being applied.
%
%\subsection{Empirical Observability Gramian}
%\label{subsec:Empgram}
%
While the observability Gramian works well for determining the observability of linear systems, analytical observability conditions for nonlinear systems quickly become intractable, necessitating simplifications
for computational tractability.
One such approach is the evaluation of nonlinear observability Gramian through linearization
and using the corresponding Jacobian matrices for linear observability analysis; however, this approach only provides
an approximation of the local observability for a specific trajectory. 
One alternative method to evaluate the observability of a nonlinear system is using the concept of observability covariance or the empirical observability Gramian \cite{Krener09}. 
This approach provides a more accurate description of a nonlinear system's observability (compared with linearization), while being computationally less expensive to apply compared with say, Lie algebraic based methods. 
To set the stage for the main contribution of the paper, 
let us first provide a brief overview of the empirical observability Gramian and how
it is used to evaluate the observability of nonlinear systems.

Consider the problem of system observability for the nonlinear system,
\begin{empheq}[left=\empheqlbrace ]{align}
 &    \dot{\vec{x}}= \vec{f} (\vec{x} , \vec{u}), \ \  \vec{x} \in \mathbb{R}^n, \ \  \vec{u} \in \mathbb{R}^p, \nonumber \\
 &	 \vec{y}=\vec{h}(\vec{x}),  \ \ \vec{y} \in \mathbb{R}^m. \label{UAffine}
\end{empheq}
The empirical observability Gramian for the system \eqref{UAffine} is constructed as follows. For $\epsilon > 0$, let $\vec{x}_0^{\pm i} = \vec{x}_0 \pm \epsilon \vec{e}_i$ be the perturbations of the initial condition and $\vec{y}^{\pm i}(t)$ be the corresponding output, where $\vec{e}_i$ is the $i^{\text{th}}$ unit vector in $\mathbb{R}^n$. For the system \eqref{UAffine}, the {\em empirical observability Gramian} $W_o$, is an $n \times n$ matrix, whose $(i,j)$ entry is 
\begin{equation}
	W_{o_{ij}} =\frac{1}{4\epsilon^2} \int_{0}^{t_f} \left(\vec{y}^{+i}(t)-\vec{y}^{-i}(t) \right)^T \left(\vec{y}^{+j}(t)-\vec{y}^{-j}(t) \right) \mathrm{d}t. \label{EmpObsGram}
\end{equation}
It can be shown that the empirical observability Gramian converges to the local observability Gramian as $\epsilon \to 0$ \cite{Krener09}. Note that the perturbation, $\epsilon$, should always be chosen such that the system stays in the region of attraction of the equilibrium point of the system. The largest singular value \cite{Singh05}, the smallest eigenvalue \cite{Krener09}, the determinant \cite{DeVries13, Serpas13}, and the trace of the inverse \cite{DeVries13} of the observability Gramian have all been used as different measures for the (nonlinear) system observability.
%if the system dynamics ($\vec{f}_0$ and $\vec{f}_i$ functions in \eqref{UAffine}) are smooth, then 

By using the definition given in \eqref{EmpObsGram}, the trace of the empirical observability Gramian can be written as
\begin{equation}
	\trace(W_o) = \frac{1}{4\epsilon^2} \int_{0}^{t_f} \sum\limits_{i=1}^n \left\| \vec{h}(\vec{x}^{+i} (t))-\vec{h}(\vec{x}^{-i}(t)) \right\|^2 \mathrm{d}t\,, \label{TraceGram}
\end{equation}
where, $\vec{x}^{\pm i} (t)$ is the trajectory corresponding to the initial condition $\vec{x}_0^{\pm i} = \vec{x}_0 \pm \epsilon \vec{e}_i$.

%For the stability analysis of nonlinear time-varying system, we must work with the following stability theorem. 
%\begin{proposition} (Theorem 4.8 from \cite{Khalil02}) \label{khalil}
%Let $\vec{x} = 0$ be an equilibrium point for a nonlinear system $\dot{\vec{x}} = \vec{f}(\vec{x},t)$, and $D \subset \mathbb{R}^n$ be a domain containing $\vec{x}=0$. Let $V: [0,\infty) \times D \rightarrow \mathbb{R}$ be a continuously differentiable function such that 
%\begin{equation}
%\begin{aligned}
%   W_1(\vec{x}) \leq V(t,\vec{x}) \leq W_2(\vec{x}) \\
%   \frac{\partial V}{\partial t} + \frac{\partial V}{\partial \vec{x}} \vec{f}(\vec{x},t) \leq 0 \,.  \label{LyapunovKhalil}
%\end{aligned}
%\end{equation}
%$\forall t \geq 0$, and $\forall \vec{x} \in D$, where $W_1(\vec{x})$ and $W_2(\vec{x})$ are continuous positive definite functions on $D$. Then, $\vec{x}=0$ is uniformly stable.
%\end{proposition}

%% =================================================================================================================== %%
\section{Problem Formulation}
\label{sec:prob}
%% =================================================================================================================== %%

Consider a system with nonlinear observation,
\begin{empheq}[left=\empheqlbrace ]{align}
 &    \dot{\vec{x}}=A \vec{x} + B \vec{u}, \ \  \vec{x} \in \mathbb{R}^n, \ \  \vec{u} \in \mathbb{R}^p, \nonumber \\
 &	 \vec{y}=\vec{h}(\vec{x}),  \ \ \vec{y} \in \mathbb{R}^m. \label{LTI_NL_h}
\end{empheq}
The linear-quadratic optimal control (known as LQR) aims at minimizing the cost functional
\begin{equation}
\begin{aligned}
& \underset{\vec{x}, \vec{u}}{\text{min}}
& & J = \int_{0}^{\infty} \left( \vec{x}^T Q \vec{x} + \vec{u}^T R \vec{u} \right) \mathrm{d}t\,,
\end{aligned}
\end{equation}
with
\begin{equation}
	Q = Q^T \succ 0, \ \  R = R^T \succ 0\,,
\end{equation}
and the linear dynamics \eqref{LTI_NL_h} represents the constraints. The solution of LQR is the linear state feedback
\begin{equation} \label{u_LQR}
	\vec{u}_{\text{LQR}} = -R^{-1} B^T P \vec{x} \,,
\end{equation}
where $P$ is the positive definite solution of the Riccati equation,
\begin{equation} \label{CARE}
	Q + A^T P + P A - P B R^{-1} B^T P = 0\,.
\end{equation}
The LQR control places the poles of $\bar{A}$ in a stable, suitably-damped locations in the complex plane, where $\bar{A}$ is the ``$A$-matrix'' of the closed loop system, and in this case $\bar{A} = A - B R^{-1} B^T P$. Picking the Lyapunov function $V(\vec{x}) = \vec{x}^T P \vec{x}$, we have,
\begin{equation}
\begin{aligned}
	V(\vec{x}) &> 0, \ \  \mbox{for all} \; \vec{x} \neq 0, \\
	\dot{V}(\vec{x}) &= \vec{x}^T P \dot{\vec{x}} + \dot{\vec{x}}^T P \vec{x} = \vec{x}^T (-Q - P B R^{-1} B^T P) \vec{x}<0, \ \  \mbox{for all} \; \vec{x} \neq 0\,.
\end{aligned}
\end{equation}
Hence, the optimal LQR control \eqref{u_LQR} asymptotically stabilizes the origin.

One potential problem may be that while the LQR problem minimizes the above cost functional, it may not result in a good practical response, e.g., when the obtained trajectory is not observable- in this case, we are not able to reconstruct the states required for state feedback. 
In fact, since the output function is nonlinear, the linearization of the output function $\vec{h}(\vec{x})$ about the optimal LQR trajectory could be unobservable (or poorly observable). In this section, our goal is modifying the LQR control such that we avoid unobservable trajectories, while ensuring that the closed loop system is stable. In order to preserve the asymptotic convergence of the closed loop state trajectory, we consider an LQR augmented control design setup. 
In such a setup the control objective considered is twofold:
\begin{itemize}
	\item Improve the observability of the system in order to avoid unobservable trajectories, and
	\item Guarantee the convergence of the resulting state trajectory to the equilibrium of the linear dynamics, i.e., the origin.
\end{itemize}

%% ======================================================================================================================= %%
%\section{Stability}
%\label{sec:stability}
%% ======================================================================================================================= %%

The suggested controller assumes the form,\footnote{We note that if in fact, the sole purpose enhancing observability is state feedback, then the proposed controller needs be implemented using the estimated state instead; see \S 5 for the example. For the purpose of this paper- however, we focus on augmented control synthesis for observability enhancement assuming the availability of the state.}
\begin{equation} \label{TV_opt_ctrl} 
	\vec{u}(\vec{x})  = -\left[ R^{-1} + S(\vec{x}) \right] B^T P \vec{x} \,,
\end{equation}
where, $S$ is a diagonal matrix with diagonal elements, $$S_{ii}(\vec{x}) = k_{1_i} e^{-\frac{k_{2_i}}{\|\vec{x}\|}} \sin^2(k_{3_i} \|\vec{x}\| + k_{4_i}), \ \ k_{1_i} \geq  0, \ \ k_{4_i} \in [0, \frac{\pi}{2}), i=1,\hdots, p,$$ and $P$ is the solution of the Riccati equation~\eqref{CARE}. 

The justification for the proposed augmentation of the state feedback obtained from LQR for enhancing observability is as follows. First, from previous research on under-sensed nonlinear systems, 
it is known that desirable trajectories for improving (nonlinear) observability are sinusoidal~\cite{lorussi2001optimal,hinson2013path}. 
%In this direction, motions corresponding to brackets of two vector fields (of interest in nonlinear observability) are generated by a switching between the appropriate controls which can be accomplished via sinusoids. 
In fact, sinusoidal control inputs achieve motions in Lie bracketing directions that are generally required for having full rank observability matrix \cite{teel1995non,murray1993nonholonomic}. Lastly, the decaying term ensures that as we approach the origin, we recover the baseline LQR controller.
%The suggested control provides the bracket of multiple control vector fields,  required for observability or generating a vector field with higher observability index.

Let us first establish the stability of the proposed control scheme.

\begin{proposition} \label{addTerm}
Given the nonlinear system with a controllable linear dynamics and nonlinear observation \eqref{LTI_NL_h}, 
%if  
%\begin{equation} \label{ctrlability}
%%\begin{aligned}
%	\text{rank} \left( \begin{bmatrix} B & AB & A^2B & \cdots & A^{n-1}B  \end{bmatrix} \right) = n\,, 
%%\end{aligned}
%\end{equation}
the control given by \eqref{TV_opt_ctrl} asymptotically stabilizes the system. 
\end{proposition}

\begin{proof} 
As the linear state dynamics is controllable, 
% \eqref{ctrlability} is satisfied, 
the control $\vec{u}_0(\vec{x}) = - R^{-1} B^T P \vec{x}$ asymptotically stabilizes the system \eqref{LTI_NL_h}. Now, consider the Lyapunov function $V(\vec{x}) = \vec{x}^T P \vec{x}$ with the control $\vec{u}(\vec{x})$ given in \eqref{TV_opt_ctrl}. Note that the Lyapunov function is positive definite. Furthermore,
\begin{equation}
	\begin{aligned}
	 	\dot{V}(\vec{x}) &= \vec{x}^T P \dot{\vec{x}} + \dot{\vec{x}}^T P \vec{x} = \vec{x}^T \left( -Q - P B R^{-1} B^T P - 2 P B S B^T P \right) \vec{x}\,.
	\end{aligned}
\end{equation}
Since $S_{ii} \geq 0 \,  (i=1, \cdots, p)$, the matrix $S$ is a positive semi-definite. Since we have $R^{-1} \succ 0$, then, $\dot{V} < 0$ for all nonzero $\vec{x}$
%\in \mathbb{R}^{>0} \times \mathbb{R}^n$, 
proving the statement of the theorem.%\footnote{Note that we have assumed $Q$ to be positive-definite.}

\end{proof} 
%% ======================================================================================================================= %%
\section{Selection of Parameters for the Augmented Controller}
\label{ch6:sec:opt}
%% ======================================================================================================================= %%

We now consider the following augmented cost functional,
\begin{equation}
\begin{aligned}
& \underset{\vec{x}, \vec{u}}{\text{min}}
& & J = \int_{0}^{t_f} \left( l_1(\vec{x},\vec{u})- w l_2(t,\vec{x},\vec{x}^{\pm 1}, \cdots, \vec{x}^{\pm n}) \right) \mathrm{d}t\,,  %J_1 - J_2 \,, % 
\end{aligned} \label{CostToGo}
\end{equation}
where,
\begin{equation} 
\begin{aligned}
	l_1(\cdot) &= \vec{x}^T Q \vec{x}+\vec{u}^TR\vec{u}\,,\\
	l_2(\cdot) &= \frac{e^{-\alpha t}} {4\epsilon^2} \sum\limits_{i=1}^n  \left\| \vec{h}(\vec{x}^{+i})-\vec{h}(\vec{x}^{-i}) \right\|^2\,, \label{l1_l2_costs} %\text{sat}_{\zeta}
%\left( \frac{1} {4\epsilon^2} \sum\limits_{i=1}^n  \left\| \vec{h}(\vec{x}^{+i})-\vec{h}(\vec{x}^{-i}) \right\|^2 \right)
\end{aligned} 
\end{equation}
and $w$ is a scalar determining the weight of the observability term ($l_2$).
The proposed cost functional consists of two parts: 
\begin{itemize}
\item The first term, $l_1(\vec{x},\vec{u})$, takes into account the control energy and the deviation from the desired steady state trajectory; in our setup we assume that $Q$ and $R$ are positive definite. 
\item The observability index, $l_2(t,\vec{x},\vec{x}^{\pm 1}, \cdots, \vec{x}^{\pm n})$, is a discounted version of the trace of the observability Gramian (given in \eqref{TraceGram}), which is a transient term, and takes into account the local observability of the system.
\end{itemize}

In the cost function \eqref{CostToGo}, the sum of the control effort (i.e., the integral of the $\vec{u}^T R \vec{u}$ term) and the deviation from the desired trajectory (i.e., the integral of the $\vec{x}^T Q\vec{x}$ term) is minimized, while maximizing an observability index. 
The index $l_1(\vec{x},\vec{u})$ determines the asymptotic behavior of the closed loop system, while $l_2(t,\vec{x},\vec{x}^{\pm 1}, \cdots, \vec{x}^{\pm n})$ specifies the desired transient behavior of the system. This cost function does not directly maximize observability.
Instead, the observability term, $l_2(t,\vec{x},\vec{x}^{\pm 1}, \cdots, \vec{x}^{\pm n})$, tunes the cost function so that the obtained optimal control makes the system more observable. We note that maximizing the index of observability in the absence of the remaining terms does not guarantee the stability of system. 

Although we have showed that the suggested control does not destabilize the system, however, our goal is that the sum of the control effort and the states cost ($l_1(\cdot)$) is the dominant term, if possible. As mentioned earlier, the main task of the observability term ($w l_2(\cdot)$) is tuning the optimal trajectory. In the next theorem, we will show there always exist the parameters ($w$ and $\alpha$) to achieve this goal.

\begin{theorem} \label{pos_cost}

Given $l_1$ and $l_2$ cost functions in \eqref{l1_l2_costs}, if the output function, $\vec{h}(\vec{x})$, is Lipschitz continuous, and the closed loop system is stable, then there always exist $w$ and $\alpha$ such that for all $t$ and nonzero $\vec{x},\vec{u},$ we have $l_1(\cdot) - w l_2(\cdot) > 0$.
\end{theorem} 

\begin{proof} 
Since the output function, $\vec{h}(\vec{x})$, is Lipschitz continuous, there exists $L$ such that
\begin{equation} \label{lipsch_cnd}
 \left\| \vec{h}(\vec{x}^{+i}(t))-\vec{h}(\vec{x}^{-i}(t)) \right\|^2 \leq L^2 \left\| \vec{x}^{+i}(t) - \vec{x}^{-i}(t) \right\|^2, \ \ \forall t\,.
\end{equation} 
As $\bar{A}$ (the ``$A$-matrix'' of the closed loop system) is stable,
$$ \vec{x}^{\pm i}(t) = e^{\bar{A} t} (\vec{x}_0 \pm \epsilon \vec{e}_i) \,.$$
Therefore,
$$ \left\| \vec{x}^{+i} - \vec{x}^{-i} \right\|^2 = \left\| 2 \epsilon e^{\bar{A} t} \vec{e}_i \right\|^2 = 4 \epsilon^2 \left\| e^{\bar{A} t} \vec{e}_i \right\|^2 \,.$$
Using the Lipschitz condition \eqref{lipsch_cnd}, we have
\begin{equation}
	\begin{aligned}
	 	l_2(\cdot) &= \frac{e^{-\alpha t}} {4\epsilon^2} \sum\limits_{i=1}^n  \left\| \vec{h}(\vec{x}^{+i})-\vec{h}(\vec{x}^{-i}) \right\|^2 \\
&\leq \frac{L^2 e^{-\alpha t}} {4\epsilon^2} \sum\limits_{i=1}^n  \left\| \vec{x}^{+i} - \vec{x}^{-i} \right\|^2 = L^2 e^{-\alpha t} \sum\limits_{i=1}^n \left\| e^{\bar{A} t} \vec{e}_i \right\|^2 = L^2 e^{-\alpha t} \left\| e^{\bar{A} t} \right\|^2_F  \,.
	 \end{aligned}
\end{equation}
Furthermore, $ \left\| e^{\bar{A} t} \right\|_F \leq \sqrt{n}  \left\| e^{\bar{A} t} \right\|_2$; as such, we have
$$ l_2(\cdot) \leq n L^2 e^{-\alpha t} \left\| e^{\bar{A} t} \right\|^2_2 \,.$$
Again by stability of $\bar{A}$, it also follows that,
$$\exists K, a >0,\ \ \left\| e^{\bar{A} t} \right\|_2 \leq K e^{-a t}, \ \ \forall t \geq 0\,.$$ 
If we pick $w$ such that
\begin{equation} \label{w_cnds}
	 	w \leq \frac{\lambda_{\min}(Q) \|\vec{x}_0\|^2}{n L^2 K^2}\,,
\end{equation}
then
$$ w l_2(\cdot) \leq \lambda_{\min}(Q) \|\vec{x}_0\|^2 e^{-(\alpha+2a)t}\,.$$
We also know that
$$ \vec{x}^T Q \vec{x} = \vec{x}_0^T e^{\bar{A}^T t} Q e^{\bar{A} t} \vec{x}_0 = \left\| Q^{\frac{1}{2}} e^{\bar{A} t} \vec{x}_0 \right\|^2 \geq \lambda_{\min}(Q) \left\| e^{\bar{A} t} \vec{x}_0 \right\|^2\,.$$
Thus,
\begin{equation}
	\begin{aligned}
	\vec{x}^T Q \vec{x} - w l_2(\cdot) &\geq \left\| Q^{\frac{1}{2}} e^{\bar{A} t} \vec{x}_0 \right\|^2 - \lambda_{\min}(Q) \|\vec{x}_0\|^2 e^{-(\alpha+2a)t} \\
	&\geq \lambda_{\min}(Q) \left( \left\| e^{\bar{A} t} \vec{x}_0 \right\|^2 - \|\vec{x}_0\|^2 e^{-(\alpha+2a)t} \right) \,.
	 \end{aligned}
\end{equation}
Now we can select $\alpha$ to satisfy the following condition:
\begin{equation} \label{alpha_cnds}
	\alpha \geq \sigma_{\max}^2(\bar{A}) - 2a \,.
\end{equation}
%where $\bar{A}$ is the ``$A$-matrix'' of the closed loop system. 
Then,
$$ \vec{x}^T Q \vec{x} - w l_2(\cdot) \geq 0 \,.$$
Using \eqref{w_cnds} and \eqref{alpha_cnds}, we now have
$$ l_1(\cdot) - w l_2(\cdot) = \{ \vec{x}^T Q \vec{x} - w l_2(\cdot)\} + \vec{u}^T R \vec{u} \geq \vec{u}^T R \vec{u} > 0 \,,$$ for nonzero $\vec{u},$ concluding the proof.
\end{proof} 
Note that since we do not know the exact values of $\bar{A}$ and $\vec{x}_0$ at the beginning, we need to have an initial guess of these values and set $w$ and $\alpha$. The values of $l_1$ and $l_2$ need to be monitored, and if $l_1 - w l_2 <0$, we need to update the values of $w$ and $\alpha$ until they meet the condition.

The optimal control $\vec{u}^*$ is hence obtained from solving
\begin{equation}
\begin{aligned}
& \underset{}{\text{minimize}}
& & \int_{0}^{t_f} \left\{ l_1(\vec{x},\vec{u})- w l_2(t,\vec{x},\vec{x}^{\pm 1}, \cdots, \vec{x}^{\pm n}) \right\} \mathrm{d}t  \\
& \text{subject to}
& & \vec{u}  = -\left[ R^{-1} + S(\vec{x}) \right] B^T P \vec{x} \,. %,\ \ t \in [t_j\ \ t_{j+1})
\end{aligned} \label{CostTran}
\end{equation}
By substituting $\vec{u}  = -\left[ R^{-1} + S(\vec{x}) \right] B^T P \vec{x}$ into \eqref{CostTran}, the cost functional can be approximated as,
\begin{equation}
   J(\vec{k}_1, \cdots, \vec{k}_4)  = \int_{0}^{t_f} \Gamma(t,\vec{x},\vec{x}^{\pm 1}, \cdots, \vec{x}^{\pm n}, \vec{k}_1, \cdots, \vec{k}_4) \mathrm{d}t \,, \label{Cost}
\end{equation}
where
\begin{equation} \label{Gamma}
\Gamma(\cdot) = \vec{x}^T \left( P B (R^{-1} + 2S +S R S) B^T P + Q \right) \vec{x} - w l_2(t,\vec{x},\vec{x}^{\pm 1}, \cdots, \vec{x}^{\pm n}) \,.
\end{equation}
Considering the above nonlinear system, 
the following optimization problem can now be considered for control synthesis:
\begin{equation} 
\begin{aligned}
& \underset{\vec{k}_1, \cdots, \vec{k}_4}{\text{minimize}}
& & J(\vec{k}_1, \cdots, \vec{k}_4) = \int_{0}^{t_f} \Gamma(t,\vec{x},\vec{x}^{\pm 1}, \cdots, \vec{x}^{\pm n},\vec{k}_1, \cdots, \vec{k}_4) \, \mathrm{d}t \label{ModCost}\\
& \text{subject to}
& & \dot{\vec{x}}=\left[ A - B (R^{-1} + S) B^T P \right] \vec{x}, \ \ \vec{x}(0) = \vec{x}_0,\\
& 
& & \dot{\vec{x}}^{\pm k}=\left[ A - B (R^{-1} + S) B^T P \right] \vec{x}^{\pm k}, \ \ \vec{x}^{\pm k}(0) = \vec{x}_0 \pm \epsilon \vec{e}_k\,, \ \ k=1,\cdots,n\,.
\end{aligned} 
\end{equation}
In the next section, we delve into designing an algorithm for
obtaining the parameters $\vec{k}_1, \cdots, \vec{k}_4$ in order to solve~\eqref{ModCost}.
\subsection{Optimization Algorithm}
In this section, an algorithm is presented to solve \eqref{ModCost} to determine the control policy for the control objective with an embedded observability index. Here, a recursive gradient-based algorithm is devised to find the solution to this optimization problem. 

Given \eqref{Cost}, we first define a new variable, $x_{n+1}$, as
\begin{equation}
    x_{n+1}(t)= \int_{0}^{t} \Gamma(\tau,\vec{x}(\tau),\vec{x}^{\pm 1}(\tau), \cdots, \vec{x}^{\pm n}(\tau),\vec{k}_1, \cdots, \vec{k}_4) \, d\tau \,.  \label{NewVar}
\end{equation}
Assuming that $\vec{x}(0)=\vec{x}_0$ is given, it is clear that $x_{n+1}(0) = 0$ and $x_{n+1}(t_f) = J(\cdot)$. 
Next, define an augmented state vector as 
$$\vec{\bar{x}}(t)= \begin{bmatrix} \vec{x}(t) \\ \vec{x}^{\pm 1}(t) \\ \vdots \\ \vec{x}^{\pm n}(t) \\ x_{n+1}(t) \end{bmatrix}\,.$$ 
Then,
\begin{equation}
\dot{\bar{\vec{x}}} = \begin{bmatrix} \dot{\vec{x}} \\ \dot{\vec{x}}^{\pm 1} \\ \vdots \\ \dot{\vec{x}}^{\pm n} \\ \dot{x}_{n+1} \end{bmatrix} = \mathbb{H}(t,\vec{\bar{x}},\vec{k}_1, \cdots, \vec{k}_4)\,; \ \ \vec{\bar{x}}(0) = \begin{bmatrix} \vec{x}_0 \\ \vec{x}_0 \pm \epsilon \vec{e}_1 \\ \vdots \\ \vec{x}_0 \pm \epsilon \vec{e}_n \\  0 \end{bmatrix}, \label{AugStateEq}
\end{equation}
where,
\begin{equation}
	\mathbb{H}(t,\vec{\bar{x}},\vec{k}_1, \cdots, \vec{k}_4) = \begin{bmatrix} 
\begin{bmatrix}  A - B (R^{-1} + S) B^T P  &0&0 \\ 0&\ddots& 0 \\ 0&0&  A - B (R^{-1} + S) B^T P \end{bmatrix} 
\begin{bmatrix} \vec{x} \\ \vec{x}^{\pm 1} \\ \vdots \\ \vec{x}^{\pm n} \end{bmatrix}
\\ \Gamma(t,\vec{x}, \vec{x}^{\pm 1}, \cdots, \vec{x}^{\pm n},\vec{k}_1, \cdots, \vec{k}_4) 
\end{bmatrix}\,,
\end{equation}
and $\Gamma(\cdot)$ is given in \eqref{Gamma}. The optimization problem can now be formulated as:
\begin{equation}
\begin{aligned}
& \underset{\vec{k}_1, \cdots, \vec{k}_4}{\text{minimize}}
& & J = x_{n+1}(t_f) \\
& \text{subject to}
& & \dot{\bar{\vec{x}}} = \mathbb{H}(t,\vec{\bar{x}},\vec{k}_1, \cdots, \vec{k}_4).
\end{aligned}
\end{equation}

In order to solve this optimization problem, the gradient $\displaystyle \frac{\partial J}{\partial \vec{k}_j} = \left. \frac{\partial x_{n+1}}{\partial \vec{k}_j} \right |_{t=t_f}$, is required~\cite{Esfahani11}.
%, a method is demonstrated for solving a similar problem setup. 
By defining $\displaystyle \bar{X}_{\vec{k}_j} = \frac{\partial \vec{\bar{x}}}{\partial \vec{k}_j}$, and applying the chain rule, we have
\begin{eqnarray}
&&    \dot{\bar{X}}_{\vec{k}_j}= \frac{\partial \mathbb{H}}{\partial \vec{\bar{x}}} \frac{\partial \vec{\bar{x}}}{\partial \vec{k}_j} +\frac{\partial \mathbb{H}}{\partial \vec{k}_j} =  \frac{\partial \mathbb{H}}{\partial \vec{\bar{x}}} \bar{X}_{\vec{k}_j} +\frac{\partial \mathbb{H}}{\partial \vec{k}_j},  \nonumber \\
&&     \bar{X}_{\vec{k}_j}(0) = \vec{0}\,, \ \ j=1,\cdots,4 \,.\label{XK_0_dot}
\end{eqnarray}
%Notice that if the system is multi-input, then $K_j$ is a matrix and the term $\displaystyle \frac{\partial \vec{\bar{x}}}{\partial K_j}$ is the derivative of a vector with respect to a matrix, which results in a higher order tensor. 
%%
%In order to address multi-input as well as single-input control problems using more conventional linear algebra, we need to have separate equations for $\displaystyle \bar{X}_{\vec{k}_{ij}} = \frac{\partial \vec{\bar{x}}}{\partial \vec{k}_{ij}}, i=1 \hdots p$, where $\vec{k}_{ij}$ is the $i^{th}$ row of the matrix $K_j$.
Here, $\frac{\partial \mathbb{H}}{\partial \vec{\bar{x}}}$ and $\frac{\partial \mathbb{H}}{\partial \vec{k}_j}$ can be computed from the given form of $\Gamma(\cdot)$. Thereby, the matrix differential equation \eqref{XK_0_dot} becomes a time varying linear ordinary differential equation that should be solved along with \eqref{AugStateEq} for $0 <t<t_f$. The last row of $\bar{X}_{\vec{k}_j}, \ \ j=1,\cdots,4$ at $t=t_f$ contains the gradient vectors which can be used for improving the parameter $\vec{k}_j$. As such, one can adopt a gradient descent algorithm for obtaining the optimal gains $\vec{k}_j^*$. Note that we can also use an estimate of the Hessian matrix to make sure that the optimal point is a minima of the objective (and not a saddle point), e.g.,\, by adopting a finite difference approximation for estimating the Jacobian presented in \cite{spall2000adaptive,alaeddini2017application}. 
\section{Illustrative Examples}
\label{sec:simulation}

In this section, we consider two examples that demonstrate the application of the results discussed in the paper. The first example pertains to a suitably constructed nonlinear system; the second example pertains to a network of linear systems over an undirected graph.

In the first example, the control synthesis procedure proposed in this paper is illustrated on a (simple) holonomic system with a nonlinear measurement. The system dynamics is given by
\begin{equation}
\begin{aligned}
   \dot{x}_1& = u_1\,, \\
   \dot{x}_2& = u_2\,.  \label{holonomic}
\end{aligned}
\end{equation}
In this case, we have a linear system $\dot{\vec{x}} = A\vec{x}+B\vec{u}$, where, $A=0$ and $B=I$. Since
$\text{rank} \left(B\right) = 2$ the system is controllable.
By choosing $Q=R=I$ and solving the algebraic Riccati equation $\displaystyle A^T P+P A-P B R^{-1} B^T P+Q=0$, we obtain $P=I$. Thus, the control policy $\displaystyle \vec{u}_{\text{LQR}} = -R^{-1}B^TP\vec{x}=-\vec{x}$ asymptotically stabilizes \eqref{holonomic}. Utilizing this state feedback controller requires full state estimation, which necessitates an observable system. 
In the meantime, if for any portion of the trajectory the state is not observable, we have no means of implementing the LQR-based
state feedback controller. 
Due to this coupling between sensing and control, we can investigate the observability of the system after applying the controller. 
Assume that the position of the vehicle is continuously measured by an omni-directional camera centered at the origin (bearing measurement). Then, the output function is given as,\footnote{We note that the measurement function in this example, $y$, is undefined in case of $x_1 = 0$. Here, we have assumed that $x_1 \neq 0, \ \  \forall t$.}
\begin{equation} \label{NL_obs}
	y = \frac{x_2}{x_1},
\end{equation}
where $y \in \mathbb{R}$ is a bearing only measurement, providing information on 
the direction of the vehicle but not on the range. 
The observability matrix is now given by
\begin{equation}
	d\mathcal{O} =\frac{\partial}{\partial \vec{x}} \begin{bmatrix} y \\ \dot{y} \\ \vdots \end{bmatrix} = \begin{bmatrix} -\frac{x_2}{x_1^2} & \frac{1}{x_1} \\ 2\frac{x_2}{x_1^3}u_1-\frac{1}{x_1^2}u_2 & -\frac{1}{x_1^2}u_1 \\ \vdots & \vdots \end{bmatrix}.
\end{equation}
The system is locally observable if the observability matrix $d\mathcal{O}$ is full rank \cite{nijmeijer1990nonlinear}. The observability matrix for the nonlinear system, however, has infinite rows. Here, we can show that there exists a set of control inputs such that the observability matrix becomes full rank. Let us consider the first two rows of the observability matrix; then we have
$$ \det(d\mathcal{O}) = \frac{1}{x_1^3} \left( u_2 -\frac{x_2}{x_1}u_1 \right) \,.$$
Hence when $u_2 \neq \frac{x_2}{x_1}u_1$, the observability matrix is full rank, and as a result, the system is observable. 

Now, if we substitute the optimal control obtained from LQR, $\displaystyle \vec{u}_{\text{LQR}} =-\vec{x}$, we have 
\begin{equation}
	\dot{y} = \frac{\dot{x}_2 x_1- \dot{x}_1 x_2}{x_1^2}= \frac{-x_1x_2+x_1x_2}{x_1^2} = 0.
\end{equation}
All other differential terms are zero ($\ddot{y} = y^{(3)} = \hdots = 0$).
Thereby by applying the control $\vec{u}_{\text{LQR}} = -\vec{x}$, the observability matrix is rank deficient, and the system is not observable with this choice of control. This phenomena is well known in computer vision and stems from the fact that we cannot ``observe" range between the camera and the vehicle from bearing measurements only. Therefore, while we have optimal controls for this type of linear system, the nonlinear measurement needs to be managed appropriately for observability.

Now we augment the observability index to the optimization problem. Assume that the instantaneous ``observability" index is given by:
$$ l_2(\cdot) = \frac{e^{-\alpha t}} {4\epsilon^2} \left( \left( y^{+1}(t)-y^{-1}(t) \right)^2 + \left( y^{+2}(t)-y^{-2}(t) \right)^2 \right) \,.$$
The modified optimal control is thus given by
\begin{equation}
   u_i(\vec{x})  = -\left(1+k_{1_i} e^{-\frac{k_{2_i}}{\|\vec{x}\|}} \sin^2(k_{3_i} \|\vec{x}\| + k_{4_i}) \right) x_i \,, \ \  i=1,2\,. \label{ModCtrl}
\end{equation}
The resulting trajectories for the initial condition $\begin{bmatrix} \pm 4, \pm 4 \end{bmatrix}$ are depicted in \cref{ObsTrajs}. The optimal controller chooses a trajectory that is not necessarily the closest path to the origin. The observability-based optimal control initially keeps the system away from an unobservable trajectory (that would have led to moving directly towards the origin), while keeping the shortest path to the origin, and eventually returning to the desired equilibrium. Although, the observability-based controller chooses the longer path, it guarantees the system observability in order to estimate the state information required for implementing the state feedback control.
\begin{figure}[!h]
   \begin{center}
  	 {\includegraphics[width=.5\textwidth]{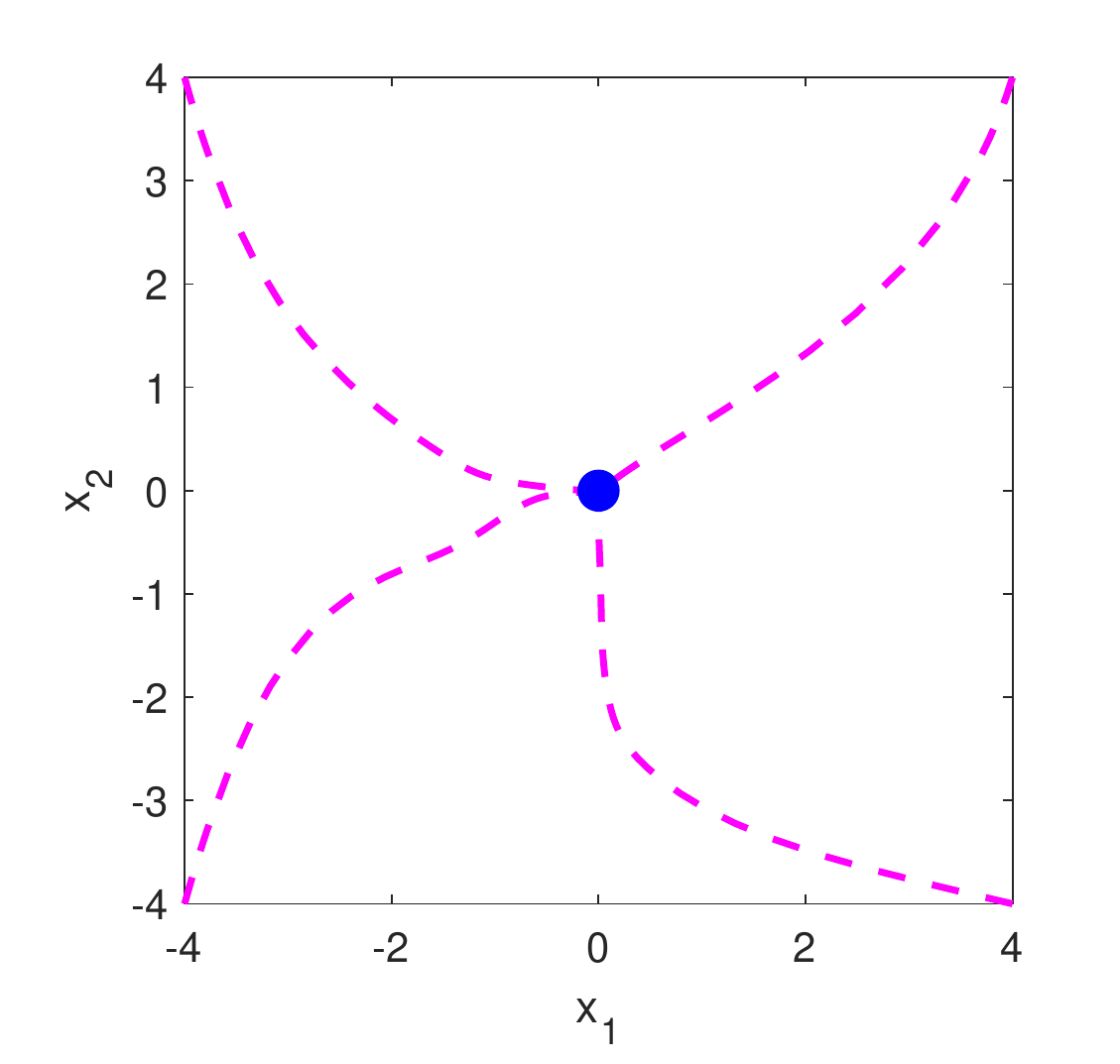}}
   \end{center}
	\vspace{-8mm}
   \caption{Trajectories for the system \eqref{holonomic} and measurement \eqref{NL_obs} with observability-based optimal control.}
   \label{ObsTrajs}
\end{figure}

Here, we compare the optimal LQR control without the observability term and the modified optimal control given in \eqref{ModCtrl} that considers the observability of the system. The optimal controls for these two scenarios are given in \cref{ObsCtrls}. As it can be seen in this figure, initially the modified optimal control deviates from the optimal LQR control, however, over time, the observability term becomes less dominant, and the absolute values of the observability-based control take the lead with respect to their corresponding LQR optimal control to satisfy the asymptotic stability condition. 
\begin{figure}[!h]
   \begin{center}
  	 {\includegraphics[width=.5\textwidth]{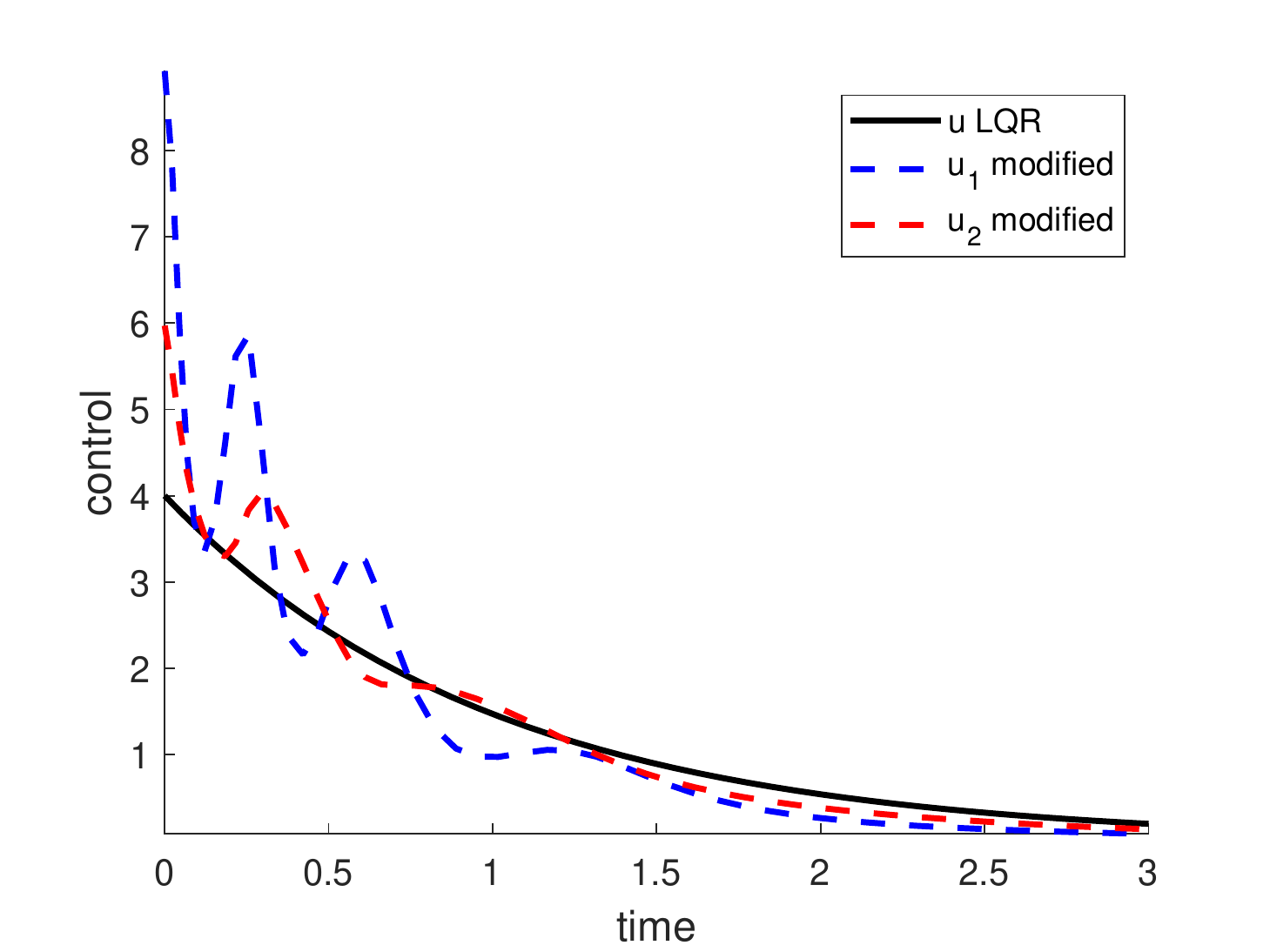}}
	 \vspace{-6 mm}
   \end{center}
   \caption{Optimal controls for the system \eqref{holonomic} and measurement \eqref{NL_obs} using LQR optimal state feedback control without (solid line) and with (dashed lines) observability-based optimal control.}
   \label{ObsCtrls}
\end{figure}

To further demonstrate the utility of the proposed observability-based control synthesis procedure, an extended Kalman filter (EKF) is considered to estimate the states of this linear system with a nonlinear measurement dynamics. The EKF is implemented in MATLAB in accordance with the presentation in~\cite{crassidis2011optimal}. Estimates are initialized to $\vec{\hat{x}}_0 = 1.25 \vec{x}_0$ with an initial covariance matrix of $P = I$. Simulation results using the augmented observability-based control are shown in \cref{estimate_optimal}. As shown in this figure, state estimates converge to the true state values for the proposed augmented feedback. In contrast, results for the LQR control are shown in \cref{estimate_LQR}. As explained above, since this system is not observable with the LQR control, the corresponding estimator does not have desirable convergence property.
\begin{figure}[!h]
  \centering
  \subfloat[]{\includegraphics[width=.49\linewidth]{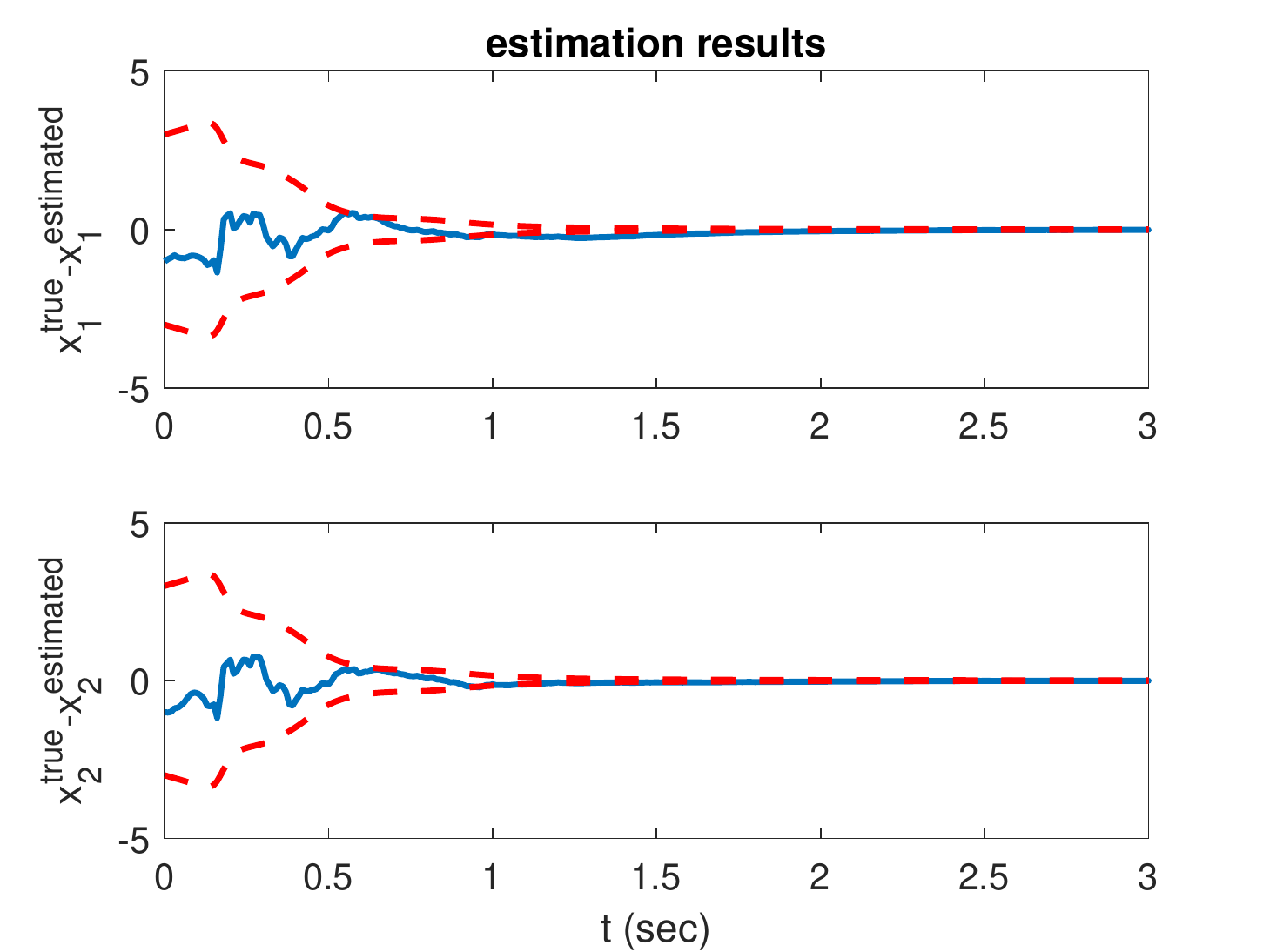} \label{EKF_result}}
	\hfill  
\subfloat[]{\includegraphics[width=.49\linewidth]{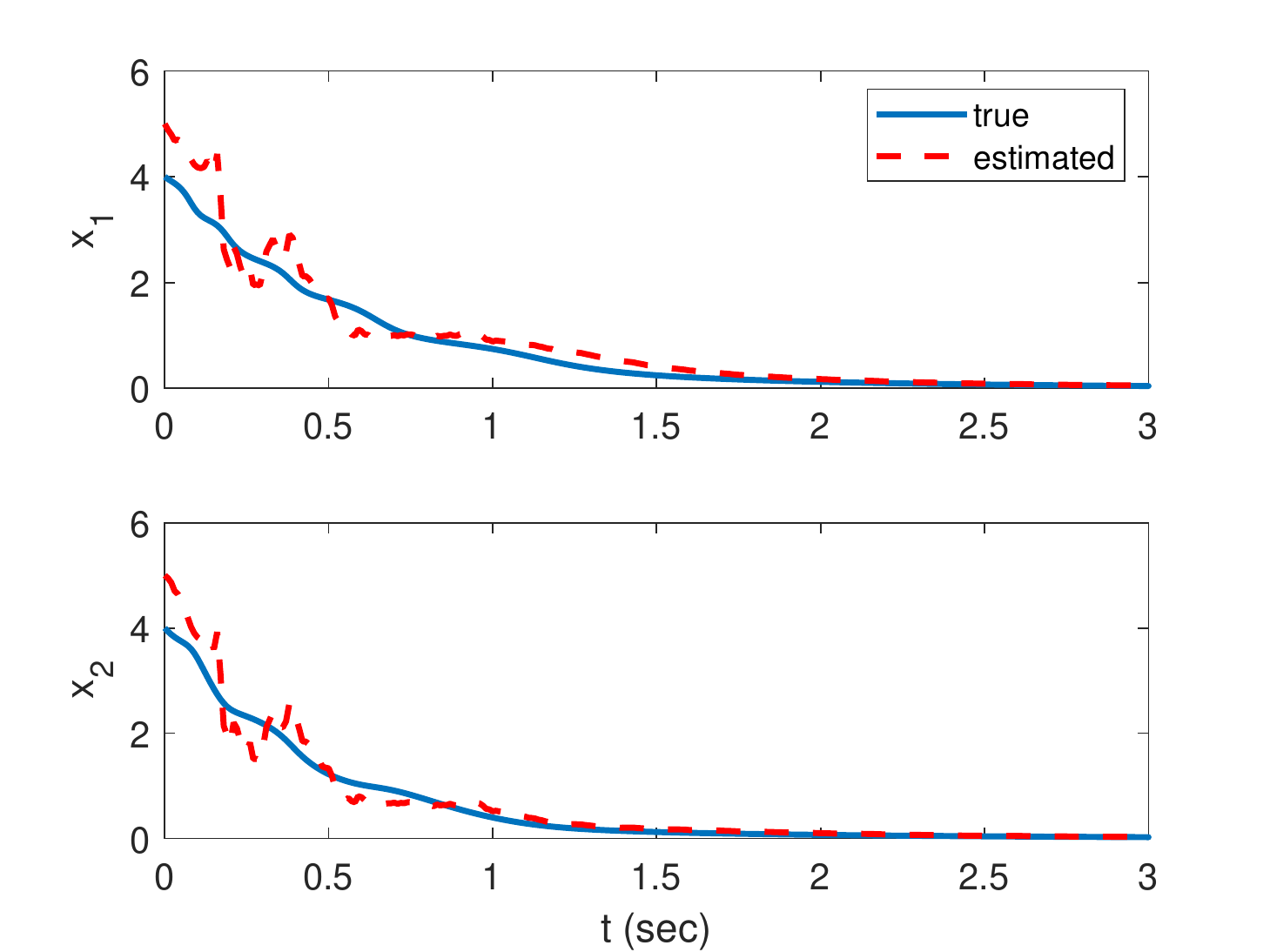} \label{EKF_states_est}}
\caption{State trajectories and estimates for the augmented control. (a) States estimation errors and $3\sigma$ bounds from the EKF estimator. (b) True and estimated states.}
\label{estimate_optimal}
\end{figure}

\begin{figure}[!h]
  \centering
  \subfloat[]{\includegraphics[width=.49\linewidth]{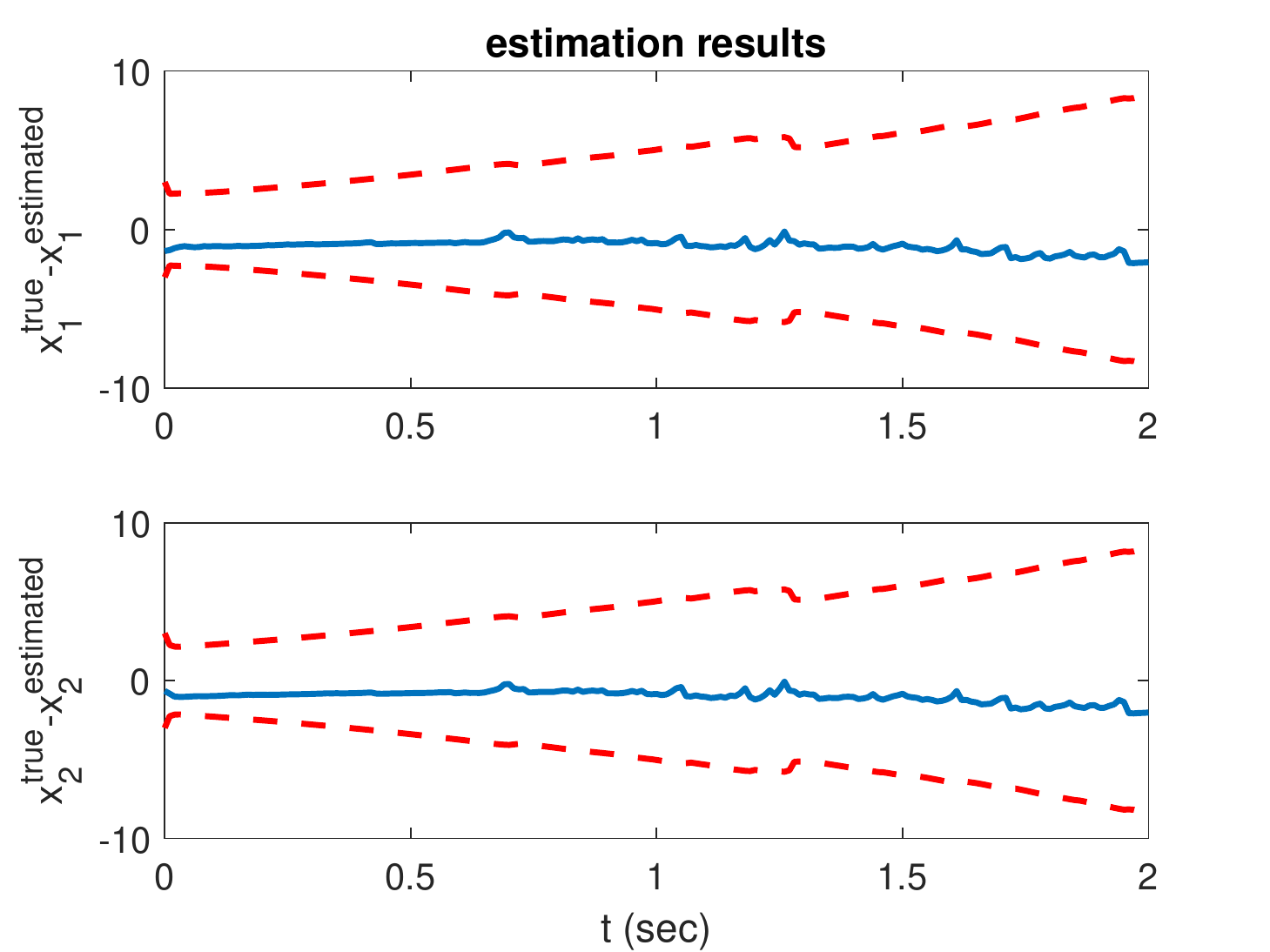} \label{EKF_result_LQR}}
	\hfill  
\subfloat[]{\includegraphics[width=.49\linewidth]{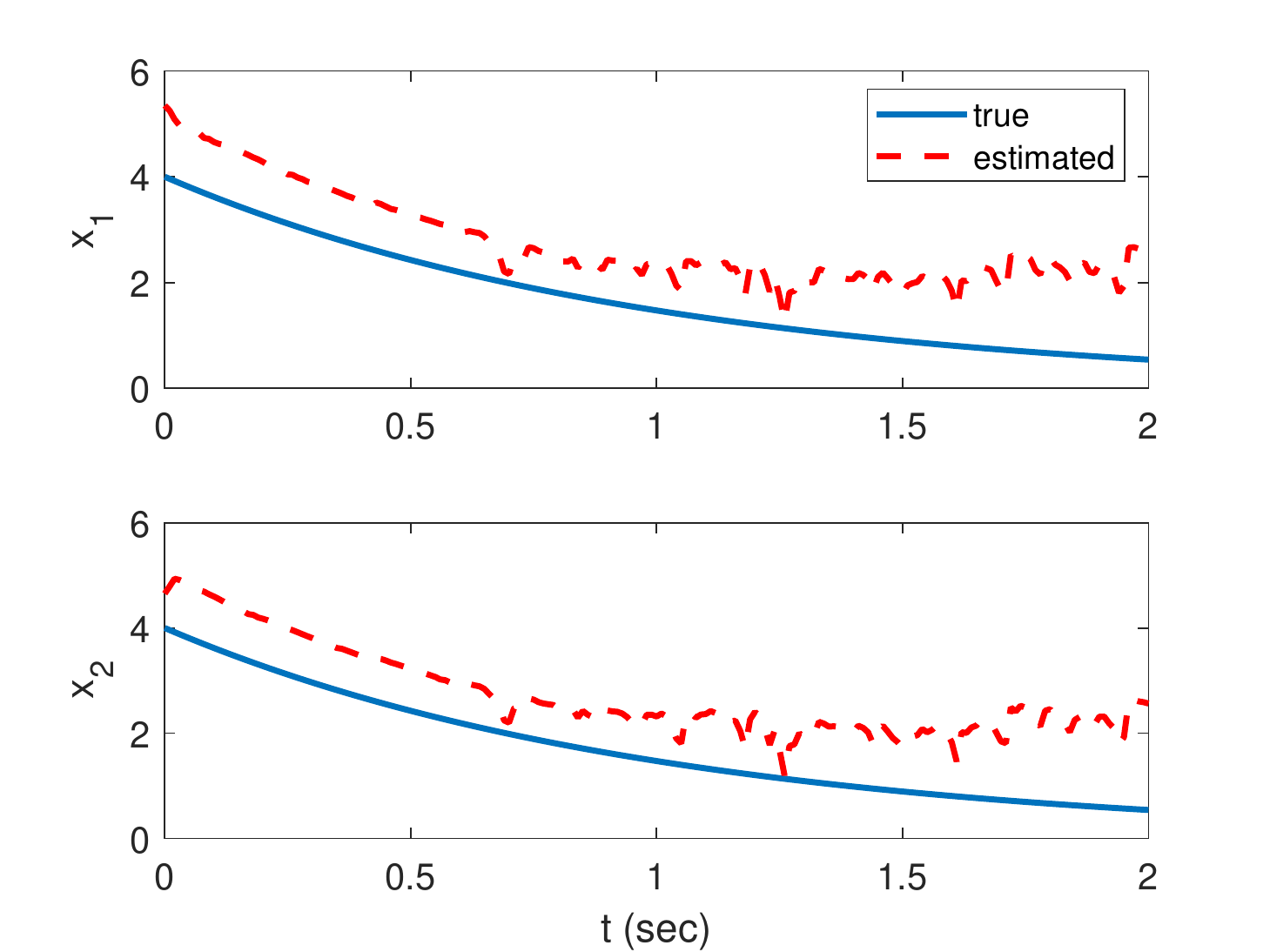} \label{EKF_states_est_LQR}}
\caption{State trajectories and estimates for the baseline LQR control. (a) States estimation errors and $3\sigma$ bounds from the EKF estimator. (b) True and estimated states.}
\label{estimate_LQR}
\end{figure}

We note that the above example can be extended to have any system matrix $A$, leading to an unobservable system as long as the control input remains linear. 
%For example, consider the model of a simple pendulum, shown in \cref{pendulum}.
%\begin{figure}[!h]
%   \begin{center}
%  	 {\includegraphics[width=.35\textwidth]{pendulum}}
%   \end{center}
%	\vspace{-1mm}
%   \caption{Simple pendulum}
%   \label{pendulum}
%\end{figure}
%If the amplitude of angular displacement is small enough that the small angle approximation $\sin{\theta} \approx \theta$ holds true, then the equation of motion reduces to the equation of simple harmonic motion, given as
%$$ \ddot{\theta} = u - \frac{b}{l}\dot{\theta}- \frac{g}{l}\theta \,.$$
%Using $\vec{x}=[\theta \ \ \dot{\theta}]^T$ for the simple example of $m=1$, $b=1$, $g=10$ and $l=1$, we have
%$$ \dot{\vec{x}} = A \vec{x} + B u\,, $$
%where
%$$ A = \begin{bmatrix} 0 & 1 \\ -10 & -1 \end{bmatrix}\,, \ \ B = \begin{bmatrix} 0 \\ 1 \end{bmatrix} \,.$$
%The optimal LQR control for this simple pendulum is $u_{\text{LQR}} = -\frac{1}{20}\left( x_1 + x_2 \right)$. Assume that we can only observe the frequency of the pendulum (nonlinear observation). Again, if we substitute the LQR control into the equation of motion, the system becomes unobservable. 
In the next example we consider the consensus problem over an undirected network of three dynamic agents on a complete graph. We assume that there is control on all nodes of this network. The corresponding dynamics can be written as~\cite{mesbahi2010graph},
$$ \dot{\vec{x}} = \begin{bmatrix} -2 & 1 & 1 \\ 1 & -2 & 1 \\ 1 & 1 & -2 \end{bmatrix} \vec{x} + \begin{bmatrix} 1  \\ 1 \\ 1 \end{bmatrix} u\,. $$
If we solve the optimal LQR control for $Q=I$ and $R=1$, we have
$$ u_{\text{LQR}} = -0.5774 x_1 - 0.5774 x_2 - 0.5774 x_3 \,.$$ 
We assume that the agents are equipped with a compass and a sensor which can measure bearing to nearby agents. In our 3-node example, the first agent observes the relative value of the other nodes with respect to its own state.
Thus, the observation of the first agent is given by,
$$ \vec{y}^1 = \begin{bmatrix} \frac{x_2}{x_1} & \frac{x_3}{x_1} \end{bmatrix}^T \,.$$
It can be shown that this network with bearing only measurement becomes unobservable under the linear time invariant LQR control. We can show that there always exists some initial condition for which the network with bearing only measurements is unobservable. Assume for example that $\vec{x}(0) = \vec{v}$, where $\vec{v}$ is an eigenvector of $A-B K$, and $K = \begin{bmatrix} 0.5774 & 0.5774 & 0.5774 \end{bmatrix}$ is the linear state feedback. Then
$\dot{\vec{x}} = \lambda \vec{x},$ where $\lambda$ is the eigenvalue associated with eigenvector $\vec{v}$. Thus, $\vec{x}(t) = e^{\lambda t} \vec{x}(0)$, and 
$$ \frac{x_j (t)}{x_k (t)} = \frac{e^{\lambda t} x_j(0)}{e^{\lambda t} x_k(0)} = \frac{ x_j(0)}{x_k(0)} = \mbox{constant}. $$
Therefore, all time derivatives of the output is zero, i.e., $\frac{d^i}{dt^i} \left( \frac{x_j}{x_k} \right) =0$ for all $i \geq 1$, and since the number of measurements for each node (which is equal to the number of neighbors of that node) is always less than the dimension of the  system (which is equal to the number of all nodes), the observability matrix cannot be full rank and agents are not able to reconstruct the states that are required for utilizing feedback control. 

In both cases mentioned above, the states can not be reconstructed using bearing measurements only. More generally, it can be shown that any system with linear dynamics can be made to be unobservable with arbitrary linear control inputs for some nonlinear measurement. Assume for example that we have a linear dynamics of the form $\dot{\vec{x}} = A \vec{x} + B \vec{u}$ with a linear control inputs $\vec{u} = K \vec{x}$ and a nonlinear observation $\vec{y}=\vec{h}(\vec{x})$, $\vec{y} \in \mathbb{R}^m$. Now let $\vec{z} = \vec{h}(\vec{x})$; if we consider the set of nonlinear observations whose time derivatives can be written as $\dot{\vec{y}} = H_1 \vec{x} + H_2 \vec{z}$, then we have
\begin{equation}
\begin{aligned}
 \begin{bmatrix} \dot{\vec{x}} \\ \dot{\vec{z}} \end{bmatrix} = \begin{bmatrix} A + B K & 0 \\ H_1 & H_2 \end{bmatrix} \begin{bmatrix} \vec{x} \\ \vec{z} \end{bmatrix} \,, \quad 
 \vec{y} = \begin{bmatrix} 0 & 1 \end{bmatrix} \begin{bmatrix} \vec{x} \\ \vec{z} \end{bmatrix} \,.
\end{aligned}
\end{equation}
For this time-invariant system, there is a convenient observability test. In particular, the rank of the observability matrix,
$$ \mathcal{O} = \begin{bmatrix} 0 & 1 \\ H_1 & H_2 \\ H_2 H_1 + H_1 (A + B K) & H_2^2 \\ H_2^2 H_1 + H_2 H_1 (A + B K) + H_1 (A + B K)^2 & H_2^3 \\ \vdots & \vdots \end{bmatrix},$$
determines whether the system is observable \cite{brockett2015finite}.
We can see that if $H_1 = 0$, then the system is unobservable for any linear state feedback gain $K$. Note that the case where $H_1 = 0$ is only one scenario where this unobservability phenomena occurs; we can find other combinations of the terms $H_1$, $H_2$ that make this system unobservable for a linear state feedback. Needless to say, the same phenomena also occurs for other classes of nonlinear measurements. 
%Furthermore, other class of nonlinear systems, and one can find many more examples of nonlinear functions for observation that does not satisfy the observability condition.
As such, the synthesis procedure detailed in this paper becomes pertinent for a large class of control systems consisting of a linear state dynamics augmented with nonlinear observations.
%%%%%%%%%%%%%%%%%%%%%%%%%%%%%%%%%%%%%%%%%%%%%%%%%%%%%%%%%%%%%%%%%%%%%%%%%%%%%%%%
\section{Conclusion}
\label{sec:conclude}

This paper is concerned with modifying the optimal control for a linear system with nonlinear measurements based on a nonlinear observability criteria. In this direction, the exponential discounted form of the empirical observability Gramian has been used for improving the local observability for this class of nonlinear systems. We proposed a modified cost function that contains a term that determines asymptotic behavior of the system (similar to the conventional LQR) in addition to a transient term responsible for maximizing a notion of nonlinear observability. Hence the augmented cost function becomes a combination of quadratic and non-quadratic terms, motivated by the desire to maximize the observability of the nonlinear system.
We then considered the stability of this augmented closed loop system; in particular,
 %it has been shown that the stability of this augmented closed loop system is guaranteed. We considered 
 we proposed an oscillatory feedback control to increase stability properties of the feedback system while also improving the observability of the nonlinear system. 
%We concerned the joint tasks of stabilizing while enabling us to improve some other aspects of the system. 
A method was then proposed that relies on superimposing a transient time-varying oscillatory term on a stabilizing controller;
it is shown that the addition of this transient term does not affect the asymptotic stability of this class of nonlinear systems.

\section*{ACKNOWLEDGMENTS}
The authors thank Bill and Melinda Gates for their active support and their sponsorship through the Global Good Fund. The research of M. Mesbahi has been supported by ONR grant N00014-12-1-1002 and NSF grant SES-1541025. Constructive suggestions and comments by the Associate Editor and reviewers of this manuscript are also acknowledged.

\section*{References}

\bibliography{citations}

\end{document}